\documentclass[10pt, conference]{IEEEtran}

\usepackage{cmap}  %
\usepackage[T1]{fontenc}  %
\usepackage[utf8]{inputenc}
\usepackage{microtype}  %

\usepackage{amsmath,amssymb,graphicx}
\usepackage{txfonts}
\usepackage{stmaryrd}
\usepackage{url}
\usepackage[english]{babel}

\usepackage{hyperref}

\usepackage{tikz}

\DeclareMathOperator{\Ai}{Ai}
\DeclareMathOperator{\erf}{erf}
\DeclareMathOperator{\EXP}{Exp}
\newcommand{\abs}[1]{\mathopen|#1\mathclose|}
\newcommand{\ttarget}{p}
\newcommand{\twork}{t}
\newcommand{\R}{\mathbb{R}}
\newcommand{\Z}{\mathbb{Z}}
\newcommand{\rnd}[1]{\left\langle #1 \right\rangle}
\newcommand{\transp}[1]{#1^{\operatorname T}}

\renewcommand{\Phi}{\mathcal{F}}
\usepackage[english, ruled, vlined, linesnumbered]{algorithm2e}
\SetAlCapSkip{\abovecaptionskip}
\SetKwComment{tc}{}{}
\SetKwRepeat{doWhile}{do}{while}
\SetKwFunction{KwImplementer}{implementer}
\SetKwFunction{KwOutput}{output}
\SetKwData{Nv}{N}
\SetKwData{Mv}{M}
\SetKwData{yv}{y}
\SetKwData{av}{a}
\SetKwData{bv}{b}
\SetKwData{tmpv}{tmp}

\newtheorem{corollary}{Corollary}
\newtheorem{lemma}{Lemma}
\newtheorem{notation}{Notation}

\newtheorem{proposition}{Proposition}
\newtheorem{theorem}{Theorem}

\newcommand{\assign}{:=}

\newcommand{\mathd}{\mathrm{d}}

\begin{document}

\setlength{\emergencystretch}{1em}
\abovedisplayskip=3pt plus 4pt
\abovedisplayshortskip=0pt plus 3pt
\belowdisplayskip=3pt plus 4pt
\belowdisplayshortskip=0pt plus 3pt
\setlength{\textfloatsep}{5pt} 
\setlength{\abovecaptionskip}{0pt}
\setlength{\belowcaptionskip}{0pt}

\title{Multiple-precision evaluation of the Airy $\Ai$ function with
reduced cancellation}
\author{
    \IEEEauthorblockN{Sylvain Chevillard\\}
    \IEEEauthorblockA{
        Inria, Apics Project-Team, Sophia Antipolis, France\\
        \href{mailto:sylvain.chevillard@inria.fr}{sylvain.chevillard@inria.fr}
    }
    \and
    \IEEEauthorblockN{Marc Mezzarobba\\}
    \IEEEauthorblockA{
        Inria, LIP (CNRS-ENS-Inria-UCBL), ENS de  Lyon, France\\
        \href{mailto:marc@mezzarobba.net}{marc@mezzarobba.net}
    }
}
\maketitle

\begin{abstract}
The series expansion at the origin of the Airy function~$\Ai(x)$ is alternating and hence problematic to evaluate for $x > 0$ due to cancellation.
Based on a method recently proposed by Gawronski,
M\"uller, and Reinhard, we exhibit two functions $F$~and~$G$, both with
nonnegative Taylor expansions at the origin, such that $\Ai(x) = G (x)
/ F (x)$. The sums are now well-conditioned, but the Taylor coefficients
of~$G$ turn out to obey an ill-conditioned three-term recurrence. We use the
classical Miller algorithm to overcome this issue. We bound
all errors and our implementation allows an arbitrary and certified accuracy,
that can be used, e.g., for providing correct rounding in arbitrary precision.
\end{abstract}


\newcommand\mightbeomitted{}

Many mathematical functions (e.g., trigonometric functions, $\erf$, Bessel functions) have a Taylor series of the form
\begin{equation}
  y(x) = x^s\,\sum_{n = 0}^{\infty} y_n\,x^{dn}, \hspace{1em} y_n \sim (-1)^n \lambda \frac{\alpha^n}{n!^{\kappa}}
  \label{eq:intro}
\end{equation}
with $d, s \in \Z$ and $\alpha, \kappa > 0$.
For large $x > 0$, the computation in finite precision arithmetic of such a sum
is notoriously prone to {\emph{catastrophic cancellation}}.
Indeed, the terms $|y_n x^n|$ are first growing before the series ``starts to
converge'' when $n^{\kappa} \ge \alpha x$. In particular, when $n^{\kappa} \approx \alpha x$,
the terms $y_n x^n$ usually get much larger than~$y(x)$. Eventually, their leading
bits cancel out while lower-order bits that actually contribute to the first
significant digits of the result get lost in roundoff errors.

This cancellation phenomenon makes the direct
computation by Taylor series impractical for large values of $x$.
Often, the function~$y(x)$ admits an asymptotic expansion as
{$x \rightarrow + \infty$}
that can be used very effectively to obtain numerical approximations when~$x$
is large, but might not provide enough accuracy (at least without resorting to sophisticated resummation methods) for intermediate values of~$x$.

In the case of the error function $\operatorname{erf}(x)$, a
classical trick going back at least to Stegun and
Zucker~{\cite{StegunZucker1970}} is to compute $\operatorname{erf}(x)$
as $G(x) / F(x)$ where $F(x) =
e^{x^2}$ and~{\cite[Eq.~7.6.2]{DLMF}}
\begin{equation}
  G(x) = e^{x^2} \operatorname{erf}(x) = \frac{2
  x}{\sqrt{\pi}}  \sum_{n = 0}^{\infty} \frac{2^n}{1 \cdot 3 \cdots(2
  n + 1)} x^{2 n} . \label{eq:std recond erf}
\end{equation}
The benefit of this transformation is that $F$~and~$G$ are power series with
nonnegative coefficients, and can thus be computed without cancellation.
Algorithms based on~\eqref{eq:std recond erf} tend to behave well in some
range $a < x < b$ where $x$~is large enough for cancellation to be problematic
but small enough to make the use of asymptotic expansions at infinity
inconvenient.
Note that the obvious way to compute $y(x) = e^{-
x}$ for $x > 0$ fits into the same framework, now with $G(x) = 1$
and $F(x) = e^x$.

Gawronski, M\"uller and
Reinhard~{\cite{GawronskiMullerReinhard2007,Reinhard2008}} provide elements to
understand where these rewritings ``come from''. They relate the amount of
cancellation in the summation of a series~\eqref{eq:intro} to the shape of the
Phragm\'en--Lindel\"of indicator of~$y$, a classical tool from the theory of
entire functions~{\cite{Levin1996}}. This description allows them to state
criteria for choosing auxiliary series suitable for the evaluation of a given
entire function in a given sector of the complex plane. They apply their
method (called the ``GMR method'' in what follows) to obtain ``reduced
cancellation'' evaluation algorithms for the error function and other related
functions in various sectors.

In this article, we are interested in the evaluation for positive~$x$ of the
Airy function $\Ai$~\cite[Chap.~9]{DLMF}.
The function $\Ai(x)$ satisfies the linear ordinary
differential equation (LODE)
\begin{equation}
  \Ai''(x) - x \Ai(x) = 0
  \label{eq:deq Ai}
\end{equation}
with initial values
\[
  \Ai(0) = A \assign 3^{-2 / 3} \Gamma (\tfrac{2}{3})^{-1}, \hspace{1em}
    \Ai'(0) = - B \assign -3^{-1 / 3} \Gamma (\tfrac{1}{3})^{-1}.
\]
The classical existence theorem for LODE with complex analytic coefficients
implies that $\Ai(x)$ is an entire function; and
solving~\eqref{eq:deq Ai} by the method of power series yields the Taylor
expansion $\Ai(x) = A f(x^3) - Bxg(x^3)$, where
\begin{align*}
  f(x) &=
  \sum_{n=0}^{{\infty}}{\frac{1{\cdot}4{\cdots}(3n-2)}{(3n)!}} x^n, &
  g(x) &=
  \sum_{n=0}^{{\infty}}{\frac{2{\cdot}5{\cdots}(3n-1)}{(3n+1)!}} x^n.
\end{align*}
Observe that while $f$~and~$g$ are easy enough to evaluate individually, the
difference~$A f(x^3) - Bxg(x^3)$ causes catastrophic cancellation when computed in
approximate arithmetic.

Using the GMR method, we derive a reduced cancellation algorithm for
computing $\Ai(x)$.
To our best knowledge, our algorithm for evaluating $\Ai(x)$ is new, and is the
most efficient multiple-precision evaluation of $\Ai(x)$
when $x$ is neither too small nor too large, while the precision is not large
enough to make methods based on binary splitting~\cite{BrentZimmermann2010} competitive.

Besides the new application, the main difference between the
present article and the work of Gawronski {\emph{et al.}} is our setting of
multiple-precision arithmetic ``\`a la
MPFR~{\cite{FousseHanrotLefevrePelissierZimmermann2007}}''. Specifically, on
the one hand, we are interested in {\emph{arbitrary precision}} arithmetic
rather than machine precision only. This makes it impossible, for instance, to
tabulate the coefficients of auxiliary functions when these turn out to be
hard to compute. Also, we are looking for {\emph{rigorous error bounds}}
instead of experimental error estimates. On the other hand, we restrict
ourselves to numerical evaluation on a half line instead of a complex sector (though, in principle, the basic ideas generalize).

This article focuses on providing a complete algorithm in the specific
case of $\Ai(x)$, $x > 0$. Yet, it should also be seen a
{\emph{case study}}, part of an effort to understand what the GMR method can
bring in the context of multiple-precision computation, and, perhaps more
importantly, how general and systematic it can be made. We discuss this last
point further in Sec.~\ref{sec:conclusion}.

The rest of this text is organized as follows.
In Sec.~\ref{sec:GMR}, we use the GMR method to choose the functions $F$~and~$G$.
Then, in Sec. \ref{sec:F}~and~\ref{sec:G}, we derive a few mathematical properties of these functions, including recurrences for their series expansions and various bounds.
Sections~\ref{sec:roundoff Analysis} to~\ref{sec:evaluation of F} contain the details of our algorithm and its error analysis.
Finally, in Sec.~\ref{sec:implementation}, we briefly describe our implementation of the algorithm.

\section{The GMR method}
\label{sec:GMR}

We now review the GMR method and apply it to obtain candidate
auxiliary series for the evaluation of the $\Ai$ function. Since the
method itself is not crucial for our results, we summarize it in intuitive
terms and refer the reader to the original
works~{\cite{GawronskiMullerReinhard2007,Reinhard2008}} for more careful
statements.

\newcommand\myfrac[2]{#1/#2}
\newcommand\axes{
  \draw[->] (0,-2/3-0.1) -- (0,4/3+0.2);
  \draw[->] (-pi-.2,0) -- (pi+.2,0);
  \begin{scope}[inner sep=1pt, outer sep=2pt]
    \foreach \x / \xtext in {
          -pi / -\pi,
          -0.666*pi / -\myfrac{2\pi}{3},
          -0.333*pi / -\myfrac\pi3,
          0.333*pi / \myfrac\pi3,
          0.666*pi / \myfrac{2\pi}{3},
          pi / \pi,
    }
      \draw (\x,1pt) -- (\x,-1pt) node[below, fill=white] {$\xtext$};
    \foreach \i in {-2, 2, 4}
      \draw (1pt,\i/3) -- (-1pt,\i/3) node[left, fill=white] {$\i/3$};
    \end{scope}
}
\tikzset{x=.64cm,y=.8cm,font=\scriptsize}

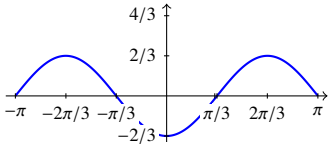
\begin{figure}
  \centerline{
    \begin{tikzpicture}
      \draw[color=blue,thick] plot[smooth,domain=-pi:pi,samples=100]
        (\x, {-2/3*cos(1.5*\x r)});
      \axes
    \end{tikzpicture}
  }
  \caption{The indicator function $h$ of $\Ai$.\label{fig:h Ai}}
\end{figure}

The starting point of the GMR method is the following observation. Let $y
(z) = \sum_{n \geq 0} y_n z^n$ be an entire function. Assume
that we have, in some intentionally vague sense,
\begin{equation}
  y(re^{i \theta}) \approx \exp(h(\theta)
  r^{\;\rho}) \label{eq:or mag y}
\end{equation}
for large~$r$. (To make things precise, we would assume that $y$ has finite
{\emph{order}} $\rho$, and that $h$~is its {\emph{indicator function}} with
respect to~$\rho$~{\cite{Levin1996}}.)

We consider the computation of~$y(z)$ in floating-point
arithmetic using its series expansion. It is well-known~\cite{Chevillard2012} that, if the sum is performed in floating-point arithmetic of precision~$t$, the relative error between $y(z)$ and the computed sum is roughly given by $2^{-t}\,(\sum_{n \ge 0} |y_n z^n|)/|y(z)|$. The sum $\sum_{n \ge 0} |y_n z^n|$ is larger than $\max_{n\ge 0} |y_n z^n|$, and usually of the same order of magnitude. Therefore, the number of significant binary digits ``lost by cancellation'' is roughly
\[ \log_2(\max_{n \geq 0}  \abs{y_n z^n}) - \log_2
   \abs{y(z)} . \]

Denote
$M(r) = \sup_{\abs{z} = r} \abs{y(z)}$
for all $r>0$.
Cauchy's formula implies $\max_n (\abs{y_n} r^n)
\leq M(r)$, and under a suitable version of
hypothesis~\eqref{eq:or mag y}, one can actually show that $\max_n (
\abs{y_n} r^n) \approx M(r)$. Hence, the loss of
precision by cancellation in the evaluation of $y(re^{i \theta}
)$ is about
\[ \log_2 \frac{M(r)}{\abs{y(re^{i \theta})
  }} \approx \left[(\max h) - h(\theta)
   \right] r^{\;\rho} . \]
For instance, when the $y_n$ all have the same complex argument, the maximum
of~$h$ is reached for~$\theta = 0$, in accordance with the fact that the sum
is optimally conditioned.

In the case of the Airy $\Ai$ function, the following asymptotic
equivalent holds as $z$~tends to complex infinity in any open sector that
avoids the negative real axis~{\cite[Eq.~9.7.5]{DLMF}}:
\begin{equation}
  \Ai(z) \sim \widetilde{\Ai}(z)
  \assign \frac{\exp(- \frac{2}{3} z^{3 / 2})}{2 \sqrt{\pi}
  z^{1 / 4}} \label{eq:Ai-tilde} .
\end{equation}
Additionally, $\Ai(x)$ is bounded for $x < 0$. Hence, we
may take $\rho = 3 / 2$ and
\begin{equation}
  h(\theta) = - \tfrac{2}{3} \cos (\tfrac32 \theta),
  \hspace{2em} - \pi < \theta \leq \pi \label{eq:h Ai}
\end{equation}
(see Figure~\ref{fig:h Ai}). The loss of precision is roughly proportional to
$1 + \cos (\frac{3}{2} \theta)$. It is minimal in the directions of fastest
growth $\theta = \pm \frac{2}{3} \pi$, and maximal for~$\theta = 0$.

If now two entire functions $y$~and~$F$ both satisfy conditions of the
form~\eqref{eq:or mag y} with the same~$\rho$ but different~$h$ (say $h_y$ and
$h_F$, respectively), we may expect that
\begin{equation}
  G(z) = F(z) y(z) \approx \exp
 ([h_y(\theta) + h_F(\theta)]
  r^{\;\rho}) . \label{eq:hG}
\end{equation}
The GMR method consists in reducing the summation of the series~$y$ for $z$ in
some given sector to that of an {\emph{auxiliary series}} $F(z)$
and a {\emph{modified series}} $G(z)$ related by~\eqref{eq:hG}.
The value of $y(z)$ is then recovered as $G(z) / F
(z)$. The auxiliary series is chosen, based on the shape
of~$h_y$, so that both $h_F$ and $h_G = h_y + h_F$ take values close to their
maximum in the sector of interest.

There may be multiple choices, and it is not clear in general which one is better, except that the coefficients of $F$~and~$G$ should be as easy to compute as possible.
Gawronski {\emph{et al.}} usually take $F(z) = \exp(
az^{\;\rho})$ and search for a value of~$a$ that makes $(\max h_G
) - h_G$ as small as possible on a whole complex sector. The choice of
exponentials as auxiliary series is not appropriate in the case of
$\Ai$, since $\exp(z^{\;\rho})$ is an entire function of~$z$
only for integer~$\rho$.

However, as we are interested in one direction only, we can easily build a
suitable auxiliary series from~$\Ai$ itself. Indeed, we may ``shift''
the indicator function of~$\Ai$ by $2 \pi / 3$ to the left or to the
right by changing $z$ to $j^{\pm 1} z$, where $j = e^{2 \pi i / 3}$. (Note
that this is {\emph{not}} the same as changing~$\theta$ to $\theta \pm
\frac{2}{3} \pi$ in~\eqref{eq:h Ai}.) When we add such a shifted
indicator to the original~$h_{\Ai}$, one of the humps of the curve cancels out with
the valley in the middle.

\begin{figure}
  \begin{tikzpicture}
    \draw[blue,thick] plot[smooth,domain=-pi/3:pi/3] (\x, {4/3*cos(1.5*\x r)});
    \axes
    \begin{scope}[color=blue, thick]
      \draw (-pi,0) -- (-pi/3,0);
      \draw (pi/3,0) -- (pi,0);
    \end{scope}
  \end{tikzpicture}
  \hfill
  \begin{tikzpicture}
    \begin{scope}[color=blue, thick]
      \draw plot[smooth,domain=-pi/3:pi/3] (\x, {2/3*cos(1.5*\x r)});
      \draw plot[smooth,domain=-pi/3:pi/3] (\x-2/3*pi, {2/3*cos(1.5*\x r)});
      \draw plot[smooth,domain=-pi/3:pi/3] (\x+2/3*pi, {2/3*cos(1.5*\x r)});
    \end{scope}
    \axes
  \end{tikzpicture}
  \caption{Indicator functions for~$F$ (left) and~$G$ (right).\label{fig:h F G}}
\end{figure}
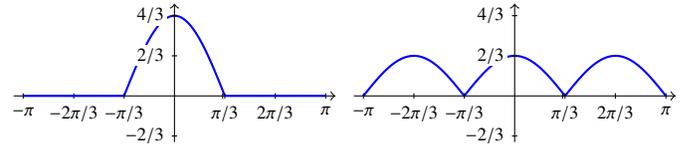

Using this idea, we set
\begin{equation}
  F(x) = \Ai(jx) \Ai(j^{- 1} x), \hspace{1em} G(x) = F(x) \Ai (x).
  \label{eq:def F G}
\end{equation}
The indicator functions of $F$~and~$G$ are pictured on Figure~\ref{fig:h F G}.
Based on their shapes, we expect that both series are optimally conditioned on
the positive real axis.
We shall prove that this is indeed the case in the
next two sections.

This second method of constructing auxiliary series seems to be new, and applies to many cases.
For instance, applying it to the error function leads to
\begin{align}
  \label{eq:new recond erf}
    F(x) &= -i \erf{i x} = \frac{2x}{\sqrt \pi} \sum_{n=0}^\infty \frac{1\cdot3\cdots(2n-1)}{(2n+1)!} 2^{n} x^{2n},
\end{align}
\begin{align}
    G(x) &= F(x) \erf x = \frac{8 x^2}{\pi} \sum_{n=0}^\infty \frac{2\cdot4\cdots(4n)}{(4n+2)!} 4^n x^{4n},
\end{align}
a slightly worse alternative to~\eqref{eq:std recond erf}.
The advantage of~\eqref{eq:std recond erf} comes from the fact that $e^{x^2}$ is faster to evaluate than~\eqref{eq:new recond erf}.

\section{The auxiliary series~$F$}
\label{sec:F}

The functions $F$~and~$G$ being chosen, we need to establish appropriate formulas to evaluate them, along with error bounds.
Much of our analysis will be based on the following simple
estimate~{\cite[Chap.~4, {\S}4.1]{Olver1997}}, where $\widetilde{\Ai}$ was defined in Eq.~\eqref{eq:Ai-tilde}.

\begin{lemma}
  \label{lem:asympt approx Ai}The Airy function~$\Ai$ satisfies
  \[ \abs{\Ai(re^{i \theta}) / \widetilde{\Ai}
    (re^{i \theta}) - 1} \leq \eta_1(\theta
    ) r^{- 3 / 2}, \hspace{2em} \abs{\theta} < \pi, \]
  where
  $\eta_1(\theta) = \frac{5}{48} (\cos \frac{\theta}{2})^{-7/2}$.
\end{lemma}

Now consider the power series expansion of~$F$ at the origin:
\begin{equation}
  F(x) = \Ai(jx) \Ai(j^{- 1} x
 ) = \sum_{n = 0}^{\infty} F_n x^n . \label{eq:series F}
\end{equation}

\begin{proposition}
  \label{prop:rec F}The coefficients~$F_n$ are positive and satisfy the
  two-term recurrence relation
  \begin{equation}
   (n + 1) (n + 2) (n + 3) F_{n + 3}
    - 2(2 n + 1) F_n = 0 \label{eq:rec F}
  \end{equation}
  with initial values
  \[
    F_0 = 3^{-4 / 3} \Gamma(\tfrac{2}{3})^{-2},
    \hspace{2em} F_1 = (2 \sqrt{3} \pi)^{-1},
  \hspace{2em} F_2 = 3^{-2 / 3} \Gamma(\tfrac{1}{3})^{-2}. \]
  
\end{proposition}

\begin{IEEEproof}
  As a general fact, if two functions $w$~and~$y$ each satisfy a
  homogeneous LODE with coefficients in~$\mathbb{Q}(x)$, then
  their product $wy$ satisfies an equation of the same class that can be explicitly computed~{\cite[Sec.~6.4]{Stanley1999}}.
  The functions $\Ai(j^{\pm 1} x)$ satisfy the same
  differential equation~\eqref{eq:deq Ai} as~$\Ai(x)$
  itself. Applying the procedure mentioned above
  to two copies of that equation yields $F^{(3)}(x)
  - 4 xF'(x) - 2 F(x) = 0$.
  
  Similarly, when an analytic function~$y$ satisfies a homogeneous LODE
  over~$\mathbb{Q}(x)$, we can compute a recurrence relation with
  coefficients in~$\mathbb{Q}(n)$ on the coefficients~$y_n$ of its power
  series expansion.
  In the case of~$F$, we get~\eqref{eq:rec F}. Finally, we
  compute the initial values $F_1, F_2, F_3$ from the first few terms of the
  Taylor expansion of $\Ai(x)$:
  \begin{align*}
    F(x) &= (A-Bjx+O(x^3))(A-Bj^{-1}x+O(x^3))\\
         &= A^2-ABx+B^2x^2+O(x^3).
  \end{align*}
  It is then apparent from~\eqref{eq:rec F} that $F_n > 0$ for all $n \in
  \mathbb{N}$.
\end{IEEEproof}

Thus, the coefficients of $F(x)$ obey a two-term recurrence
whose coefficients do not vanish for $n \geq 0$. This allows one to
compute them in a numerically stable way (see Sec.~\ref{sec:evaluation of F}).

\section{The modified series $G$}
\label{sec:G}

Recall that $ G(x) = \Ai(x) \Ai(jx) \Ai(j^{- 1} x)$, and set
\begin{equation}
  \tilde{G}(x) = \widetilde{\Ai}(x) 
  \widetilde{\Ai}(jx)  \widetilde{\Ai}(j^{-
  1} x) = \frac{\exp(\frac{2}{3} x^{3 / 2})}{8 \pi^{3 /
  2} x^{3 / 4}} . \label{eq:asy G}
\end{equation}

\begin{proposition}
  \label{prop:rec G}
  The function $G$ is an entire function with power series
  expansion of the form $G(z) = \sum_{n = 0}^{\infty} G_n z^{3 n}$.
  The coefficient sequence $(G_n)_{n \in \mathbb{N}}$ is
  determined from its first terms
  \begin{align*}
    G_0 &= A^3={\frac{1}{9{\Gamma}(2/3)^3}}, &
    G_1 &= {\frac{A^3}{2}}-B^3={\frac{1}{18{\Gamma}(2/3)^3}}-{\frac{1}{3{\Gamma}(1/3)^3}}
  \end{align*}
  by the recurrence relation
  \begin{equation}
    \label{eq:rec G}
    (n+1)(3n+4)(3n+5)(n+2) G_{n + 2} - 10(n+1)^2  G_{n+1} + G_n = 0.
  \end{equation}
\end{proposition}

\begin{IEEEproof}
  First, observe that $G(jz) = G(z)$ and $G(
  \bar{z}) = \overline{G(z)}$, so that the Taylor
  expansion of~$G$ at the origin is a power series in~$z^3$ with
  real coefficients. The same routine reasoning as in the proof of
  Prop.~\ref{prop:rec F} yields the LODE
  \[ G^{(4)}(x) - 10 xG''(x) - 10 G' (x) + 9 x^2 G(x) = 0, \]
  and from there the recurrence~\eqref{eq:rec G}. The coefficients
  of~\eqref{eq:rec G} do not vanish for $n \geq 0$, so that the
  sequence~$G_n$ is indeed determined by $G_0$, $G_1$, and \eqref{eq:rec G}.
\end{IEEEproof}

Nonzero solutions of~\eqref{eq:rec G} decrease roughly as $n!^{- 2}$ for
large~$n$.
Setting~$c_n = n!^2 G_n$ yields the
``normalized'' recurrence
\begin{equation}
  \frac{(3 n + 4) (3 n + 5)}{(n + 1) 
 (n + 2)} c_{n + 2} - 10 c_{n + 1} + c_n = 0. \label{eq:norm rec
  G}
\end{equation}
Letting $n$~go to infinity in the coefficients of~\eqref{eq:norm rec G}, we
get a limit recurrence with constant coefficients whose
characteristic polynomial $9 \alpha^2 - 10 \alpha + 1$ has two roots of
distinct absolute value, namely $\frac{1}{9}$~and~$1$. By the Perron--Kreuser
theorem~{\cite[Theorem~B.10]{Wimp1984}}, it follows that any solution~$(
v_n)_{n \in \mathbb{N}}$ of~\eqref{eq:rec G} satisfies
$v_{n+1}/v_n \sim \alpha n^{-2}$
with either $\alpha = 1$ or $\alpha = \frac{1}{9}$. Solutions $(v_n
)$ such that $v_{n + 1} / v_n \sim \frac{1}{9n^2}$ are called
{\emph{minimal}} and form a linear subspace of dimension~$1$ of the solutions of~\eqref{eq:rec G}.

We shall prove that~$(G_n)$ actually is such a minimal solution of~(\ref{eq:rec
G}). But our analysis uses a bit more than the rough estimate $G_n \approx n!^{- 2} 9^{-n}$. Prop.~\ref{prop:bound on Gn} below
provides a more precise estimate which implies the minimality.
Before turning to it, we recall a standard bound on the tails of incomplete Gaussian integrals~{\cite[Eq.~7.12.1]{DLMF}} and state a second technical lemma.

\begin{lemma}
  \label{lem:erf}The complementary error function
  $\operatorname{erfc} x = 1 - \operatorname{erf} x = \frac{2}{\sqrt{\pi}} 
     \int_x^{\infty} e^{- t^2} \mathd t $
  satisfies
  $ 0 \leq \operatorname{erfc} x \leq \frac{1}{\sqrt{\pi} x} e^{- x^2} $
  for $x > 0$.
\end{lemma}

\begin{lemma} \label{lem:int} \mightbeomitted
  \label{lem:I2-I1}The expression
  \[ I(r) = r^{9 / 8}  \int_{r^{- 5 / 8}}^{\pi / 3}
     e^{\frac{2}{3} r^{3 / 2} (- 1 + \cos \frac{3 \theta}{2})}
     \mathd \theta \]
  satisfies $0 \leq I(r) \leq 0.51$ for all $r
  \geq 10$.
\end{lemma}

\begin{IEEEproof}
  We first use the inequality $\cos (\frac{3}{2} \theta) \leq 1 -
  \frac{9}{10} \theta^2$ (valid for $0 \leq \theta \leq \pi / 3$)
  followed by the change of variable $\varphi = \sqrt{3/5}\,r^{3 / 4}\,\theta$ to get
  \[ I(r) \leq r^{9 / 8}  \int_{r^{- 5 / 8}}^{\pi / 3} e^{-
     \frac{3}{5} r^{3 / 2} \theta^2} \mathd \theta = \frac{\sqrt{5}}{\sqrt{3}}\, r^{3 / 8} 
     \int_{a(r)}^{b(r)} e^{-
     \varphi^2} \mathd \varphi, \]
  where $a(r) = \sqrt{3/5}\,r^{1 / 8}$ and $b(r) =
  \frac{\pi}{ \sqrt{15}} r^{3 / 4}$. Now we have
  \[ \int_{a(r)}^{b(r)} e^{-
     \varphi^2} \mathd \varphi \leq \int_{a(r
    )}^{\infty} e^{-\varphi^2} \mathd \varphi = \frac{\sqrt{\pi}}{2}
      \operatorname{erfc}\left(\frac{\sqrt{3}
     r^{1 / 8}}{\sqrt{5}}\right), \]
  and, by Lemma~\ref{lem:erf},
  $ I(r) \leq \frac{5}{6} r^{1 / 4} e^{- \frac{3}{5} r^{ 1 / 4}}$.
  When $r \geq 10$, the last bound is decreasing and hence less than~$0.51$.
\end{IEEEproof}

The following proposition is the main result of this section and the starting point of much of the error analysis that follows.
Though somewhat technical, the proof is mostly routine.

\begin{proposition}
  \label{prop:bound on Gn}The sequence~$(G_n)$ satisfies
  \[ G_n \sim \gamma_n = \frac{1}{4 \sqrt{3} \pi 9^n n!^2}, \hspace{1em} n
     \rightarrow \infty, \]
  with relative error
  $ \abs{G_n / \gamma_n - 1} \leq 2.2 n^{- 1 / 4} $ \hspace{1em}($n \geq 1$).
\end{proposition}

\begin{notation} \label{not:E}
  We write
  $ E(\varepsilon_1, \ldots, \varepsilon_n) = \sum_{\varnothing
     \neq I \subset \left\llbracket 1, n \right\rrbracket} \prod_{i \in I}
     \varepsilon_i$,
  so that, if $\abs{\theta_i} \leq \varepsilon_i$ for all $i$, then $\prod_{i=1}^n (1+\theta_i) = 1+\theta$ with $\abs{\theta} \leq E(\varepsilon_1, \ldots,
  \varepsilon_n)$.
  Note that we obviously have $E(\alpha\,\varepsilon_1, \ldots, \alpha\,\varepsilon_n) \le \alpha\,E(\varepsilon_1, \ldots, \varepsilon_n)$ for $0 \le \alpha \le 1$.
\end{notation}

\begin{IEEEproof}
  The proof is a standard application of the {\emph{saddle point method}}
  \cite[{\S}VIII.3]{FlajoletSedgewick2009}, which we work out in
  some detail in order to get an explicit error bound.

  We fix $n > 10$. We shall write $z = r e^{i\theta}$ in the following.
  The method prescribes to choose $r = (3 n + 3 / 4)^{2 / 3}$, so that
  \begin{equation}
    \frac{\tilde{G}(z)}{z^{3 n}} = \frac{\tilde{G}(r)}{r^{3 n}} e^{O((z - r)^2)} \label{eq:saddle cond}
  \end{equation}
  (i.e., $\frac{\mathd}{\mathd z} \log (\tilde{G}(z) z^{- 3
  n})_{\left|z = r \right.} = 0$).

  Now, guided by~\eqref{eq:saddle cond} and $G(z) \approx \tilde{G}(z)$, we write
  \[ G_n = \frac{1}{2 \pi i}  \oint \frac{G(z)}{z^{3 n + 1}}
     \mathd z = \frac{3}{2 \pi}  \frac{\tilde{G}(r)}{r^{3 n}}
     I_1, \quad
     I_1 = \int_{- \pi / 3}^{\pi / 3} \frac{G(re^{i \theta})}{\tilde{G}(r)} e^{- 3 ni \theta} \mathd \theta . \]
  Most of the weight of $I_1$ is concentrated around $0$. We set
  \begin{equation}
    \theta_0 = r^{- 5 / 8} =(3 n + 3 / 4)^{- 5 / 12} .
    \label{eq:theta0}
  \end{equation}
  Since $n > 10$, we remark that $r \geq 10$ and $0 < \theta_0 < 1 / 4 < \pi / 3$.
  We further let
  \[
  I_2 = \int_{-{\theta}_0}^{{\theta}_0}{\frac{G(re^{i{\theta}})}{\tilde{G}(r)}}e^{-3ni{\theta}}{\mathd}{\theta}.
  \]

  We first bound the error between $I_1$ and $I_2$. When $\theta \in [0,\pi / 3]$,
  using Lemma~1 and the connection formula~{\cite[Eq. 9.2.12]{DLMF}}
  \[ \Ai(z) + j \Ai(jz) + j^{- 1}
     \Ai(j^{- 1} z) = 0, \]
  we see that {$\abs{G(z) / \tilde{G}(z)}$} is
  bounded by
  \begin{align*}
    \left|\tfrac{\Ai(z)}{\widetilde{\Ai}(z)}\right| &
    \left|\tfrac{\Ai(j^{-1}z)}{\widetilde{\Ai}(j^{-1}z)}\right|
    \left(
        \left|\tfrac{\widetilde{\Ai}(z)}{\widetilde{\Ai}(jz)}\right|
        \left|\tfrac{\Ai(z)}{\widetilde{\Ai}(z)}\right|
    +
    \left|\tfrac{\widetilde{\Ai}(j^{-1}z)}{\widetilde{\Ai}(jz)}\right|
    \left|\tfrac{\Ai(j^{-1}z)}{\widetilde{\Ai}(j^{-1}z)}\right|\right)\\
    \leq&
    (1+{\eta}_1({\theta})r^{-3/2})
    (1+{\eta}_1({\theta}-{\tfrac{2{\pi}}{3}})r^{-3/2})\\
    &
    \left(e^{-{\frac{4}{3}}\cos{\frac{3{\theta}}{2}}}(1+{\eta}_1({\theta})r^{-{\frac{3}{2}}})+1+{\eta}_1({\theta}-{\tfrac{2{\pi}}{3}})r^{-{\frac{3}{2}}}\right).
  \end{align*}
  It follows that $\abs{G(z)} \leq 2.1 | \tilde{G}(z) |$ when $\abs{z} \geq 10$.
  By symmetry, this inequality holds for $- \pi / 3 \leq \theta
  \leq 0$ too.
  Finally, we have
  \begin{align*}
    \left|I_1-I_2\right|
    &\leq 2\int_{{\theta}_0}^{{\pi}/3}{\frac{\abs{G(re^{i{\theta}})}}{\tilde{G}(r)}}{\mathd}{\theta}
    \leq 4.2\int_{{\theta}_0}^{{\pi}/3}{\frac{|\tilde{G}(re^{i{\theta}})|}{\tilde{G}(r)}}{\mathd}{\theta} \\
    &= 4.2\int_{{\theta}_0}^{{\pi}/3}e^{{\frac{2}{3}}r^{3/2}(-1+\cos {\frac{3{\theta}}{2}})}{\mathd}{\theta}
  \end{align*}
  and hence, using Lemma~\ref{lem:I2-I1},
  \begin{equation}
    \abs{I_1 - I_2} \leq 2.15 r^{- 9 / 8} \label{eq:bound B}
  \end{equation}

  We now have to estimate $I_2$. We write
  \[
  I_2 = \int_{-\theta_0}^{{\theta}_0} \frac{G(re^{i{\theta}})}{\tilde{G}(re^{i\theta})} \cdot \frac{\tilde{G}(re^{i{\theta}})}{\tilde{G}(r)} e^{-3ni{\theta}}{\mathd}{\theta}.
  \]
  On the one hand, Lemma~\ref{lem:asympt approx Ai} gives
  $  G(z) = \tilde{G}(z)(1+\delta(\theta))$ with
  $ |\delta(\theta)| \le \eta_2(\theta)r^{-3/2}$ where $\eta_2(\theta) = E(\eta_1(\theta),
  \eta_1(\theta + \frac{2 \pi}{3}), \eta_1(\theta - \frac{2 \pi}{3}))$.
  Note that this bound increases with~$\abs{\theta}$.

  On the other hand, thanks to~\eqref{eq:saddle cond}, we can write
  \begin{align*}
    \tilde{G}(re^{i{\theta}})e^{-3ni{\theta}}
    = \tilde{G}(r)e^{{\frac{2}{3}}r^{3/2}(\exp({\frac{3}{2}}i{\theta})-1-{\frac{3}{2}}i{\theta})}
    = \tilde{G}(r)e^{(-{\frac{3}{4}}{\theta}^2+u({\theta}))r^{3/2}}
  \end{align*}
  where
  \begin{equation}
    \abs{u(\theta)} = \tfrac{2}{3}\left|e^{\frac{3}{2} i \theta} -
    1 - \tfrac{3}{2} i \theta + \tfrac{9}{8} \theta^2\right| \leq
    \tfrac{3}{8} \abs{\theta}^3 . \label{eq:bound u}
  \end{equation}
  Let
  \[ v(\theta) = \frac{G(re^{i \theta}) e^{- 3 in
     \theta}}{\tilde{G}(r) \exp(- \frac{3}{4} r^{3 / 2}
     \theta^2)} - 1. \]
  Since $\theta_0^3 r^{3 / 2} = r^{- 3 / 8}$ by the
  choice~\eqref{eq:theta0}, using~\eqref{eq:bound u} and the
  inequality $\abs{e^z-1} \leq \abs{z} e^{\abs{z}}$, we
  have, for $\abs{\theta} \leq \theta_0$ and $r \ge 10$:
  \[ |v(\theta)| \leq \frac{3}{8}
     r^{- 3 / 8} e^{\frac{3}{8} r^{- 3 / 8}} \leq
     0.44 r^{- 3 / 8}. \]
  To sum up, we can now rewrite $I_2$ as
  \[
  I_2 = \int_{-\theta_0}^{{\theta}_0} (1+\delta(\theta))(1+v(\theta))\exp\left(-{\frac{3}{4}}r^{3/2}{\theta}^2\right){\mathd}{\theta}.
  \]

  Now, define
  \begin{align*}
    I_3 &= \int_{-{\theta}_0}^{{\theta}_0}\exp\left(-{\frac{3}{4}}r^{3/2}{\theta}^2\right){\mathd}{\theta},&
    I_4 &= {\frac{2\sqrt{{\pi}}}{\sqrt{3}}}r^{-3/4}.
  \end{align*}
  When $\theta \in [-\theta_0, \theta_0]$, we have
  \[
    \left|(1+\delta(\theta))(1+v(\theta))-1\right| \leq
    E({\eta}_2({\theta}_0)r^{-3/2},0.44r^{-3/8}) \leq 0.95r^{-3/8},
  \]
  so we get
  \begin{equation}
    \abs{I_2 - I_3} 
    \leq 0.95 r^{- 3 / 8} I_3 . \label{eq:bound C}
  \end{equation}
  Finally, $I_3$ is an incomplete Gaussian integral:
  \[ I_3 = \frac{2 \sqrt{\pi}}{\sqrt{3} r^{3 / 4}} \operatorname{erf}(
     \tfrac{1}{2} \sqrt{3} r^{3 / 4} \theta_0) = I_4 \operatorname{erf}(
     \tfrac{1}{2} \sqrt{3} r^{1 / 8}) . \]
  Lemma~\ref{lem:erf} yields
  \begin{equation}
    \abs{I_3 / I_4 - 1} \leq \frac{2 r^{- 1 / 8} e^{-
    \frac{3}{4} r^{1 / 4}}}{\sqrt{3 \pi}} \leq 0.66 r^{- 1 / 8} e^{-
    \frac{3}{4} r^{1 / 4}} . \label{eq:bound D}
  \end{equation}
  Putting together (\ref{eq:bound B}, \ref{eq:bound C},
  \ref{eq:bound D}), we obtain $I_1 = I_4(1+\epsilon)$ with
  \begin{align*}
    \left|\epsilon\right|
    &\leq \frac{\left|I_1-I_2\right|}{I_4} + \left|{\frac{I_2I_3}{I_3I_4}}-1\right|\\
    &\leq 1.06r^{-3/8} + E(0.95r^{-3/8}, 0.66r^{-1/8}e^{-{\frac{3}{4}}r^{1/4}})\\
    &\leq 2.45r^{-3/8} \leq 1.9 n^{-1/4}.
  \end{align*}

  Now, we can write $G_n = \frac{3}{2\pi}\cdot\frac{\tilde{G}(r)}{r^{3n}}\,I_4(1+\epsilon) = \gamma_n\,H\,(1+\epsilon)$ with
  \[
    H := {\frac{\sqrt{3}}{\sqrt{\pi} r^{3/4} {\gamma}_n}} {\frac{\tilde{G}(r)}{r^{3n}}}\\
    = {\frac{n!^2e^{2n}}{n^{2n}2{\pi}n}}{\frac{\sqrt{e}}{(1+{\frac{1}{4n}})^{2n+1}}}.
  \]
  Stirling's formula in the form~\cite{Robbins1955}
  $\left| \frac{n!e^n}{n^n  \sqrt{2 \pi n}} - 1 \right| \leq s(n) =  e^{\frac1{12n}}-1$,
  combined with the bound
  $ \left| \frac{\sqrt{e}}{(1 + 1 /(4 n))^{2 n}}
     - 1 \right| \leq \frac{1}{16 n} $,
  yields
  \begin{equation}
    \abs{H - 1} \leq E\left(s(n), s(n),
    \tfrac{1}{16 n}, \tfrac{1}{4 n}\right) \leq \frac{1}{2 n},
    \qquad n \geq 11.
    \label{eq:bound A}
  \end{equation}
  It follows that
  \[
    \left|{\frac{G_n}{{\gamma}_n}}-1\right|
    \leq E(0.5n^{-1},1.9n^{-1/4})\\
    \leq 2.2n^{-1/4},
    \hspace{2em} n \geq 11,
  \]
  One easily checks that this bound is valid for
  $1 \leq n \leq 10$.
\end{IEEEproof}

Note that the above bound is not the best we can get by this method.
Indeed, by choosing the exponent of~$r$ in~\eqref{eq:theta0} closer to $-3/4$, we obtain~$|G_n/\gamma_n-1|=O(n^{-1/2+\epsilon})$ for any $\epsilon>0$. This comes at the price of a larger constant factor and thus more terms to check separately.

A first consequence of Prop.~\ref{prop:bound on Gn} is that the series~$G(x)$ has nonnegative coefficients, as stated below.
We also deduce several other technical results that will be used to bound various error terms.
Recall that $c_n = n!^2 G_n$, and let $\tau = 3/20$.

\begin{corollary}
    \label{cor:positivity}  The sequence $(c_n)$ satisfies $0 \le c_{n+1}/c_n \le \tau$ for all $n \geq 0$. Accordingly, we have $0 < G_{n+1}/G_n \le \tau (n+1)^{-2}$.
    In particular, the~$G_n$ are positive.
\end{corollary}

\begin{IEEEproof}
 Prop.~\ref{prop:bound on Gn} implies that $c_n = (1+\theta_n)/(4\sqrt{3}\pi 9^n)$ with $|\theta_n| \le 2.4 n^{-1/4}$. When $n \ge n_0 = 47610$, we have $|\theta_n| \le 0.148936$. This implies $c_n \ge 0$ and
  $c_{n+1}/c_n \le (1/9)(1+\theta_{n+1})/(1+\theta_n) \le \tau$.
  The inequalities up to $n=n_0$ are checked by computing the corresponding terms with interval arithmetic.
\end{IEEEproof}

\begin{corollary} \mightbeomitted
  \label{cor:upper bound on Gn}  For any $n \ge 1$, it holds that $G_n \le (e/(3n))^{2n}$.
\end{corollary}
\begin{IEEEproof}
  Follows from Proposition~\ref{prop:bound on Gn} using $n! \ge (n/e)^n$.
\end{IEEEproof}

\begin{lemma}
  \label{lem:remainder G}
  When $x>0$ and $N+1 \ge \sqrt{2\tau} x^{3/2}$, we have
  \[ \sum_{n=0}^{N} G_n x^{3n} \le G(x) \le \sum_{n=0}^{N-1} G_n x^{3n} + 2G_N x^{3N}. \]
\end{lemma}
\begin{IEEEproof}
  The first inequality is obvious as $G_n \ge 0$. To show the second one, observe that a fortiori $\tau x^3/(n+1)^2 \le \frac12$ for all $n \ge N$. Using Corollary~\ref{cor:positivity}, we deduce $G_{n+1} x^3/G_n \le \frac12$, and hence $\sum_{n \ge N} G_n x^{3n} \le G_N x^{3N} (1 + \frac12 + \frac14 + \cdots) = 2G_N x^{3N}.$
\end{IEEEproof}

\begin{lemma} \mightbeomitted
  \label{lem:bounds G}The function~$G$ satisfies
  \[ 0.01 e^{\frac{2}{3} x^{3 / 2}} x^{- 3 / 4} \leq G(x)
     \leq 0.04 e^{\frac{2}{3} x^{3 / 2}} x^{- 3 / 4} \]
  for all $x \geq 1/2$.
\end{lemma}
\begin{IEEEproof}
  Let $\mu(x) = E (\eta_1(0) x^{- 3 / 2}, \sigma x^{- 3 / 2}, \sigma x^{- 3 / 2})$ where $\sigma = \eta_1(\frac{2 \pi}{3}) = \frac{5}{6} \sqrt{2}$.
  Lemma~\ref{lem:asympt approx Ai} implies
  $ \abs{G(x)/\tilde{G}(x) - 1} \leq \mu(x)$
  for $x \geq 0$. We can check that $\mu(x)$ is a decreasing function and
  \[ 0.01 \leq \frac{1 - \mu(3)}{8 \pi^{3 / 2}},
     \hspace{2em} \frac{1 + \mu(3)}{8 \pi^{3 / 2}} \leq
     0.04, \]
  whence the desired inequality for~$x \geq 3$.
  For $0.5 < x \le 3$, using the bounds
  $0
  \leq e^x - (1 + x + \frac12 x^2 + \frac16 x^3 + \frac{1}{24} x^4)
  \leq \frac{5}{48} x^4$
  (valid for $0 < x \le \frac23 3^{3/2}$) and Lemma~\ref{lem:remainder G}, we are reduced to checking explicit polynomial inequalities.
\end{IEEEproof}

\section{Roundoff error analysis}
\label{sec:roundoff Analysis}

We now turn to the floating-point implementation of the functions $F(x)$ and $G(x)$.
To make the algorithm rigorous, we will use classical techniques of error analysis that we briefly recall here.
We refer the reader, e.g., to Higham~\cite{Higham}, for proofs and complement of information.

We suppose that the precision of the floating-point format is $\twork$~bits and that the exponent range is unbounded (in case it is bounded, it would probably be possible to rescale $F(x)$ and $G(x)$ by the same factor, to make them representable without changing the ratio $F(x)/G(x)$).

\begin{notation}
  \label{not:EXP}
  For $x \neq 0$, we denote $\EXP(x) = \lfloor \log_2 |x| \rfloor + 1$, so that $2^{\EXP(x)-1} \le |x| < 2^{\EXP(x)}$.
\end{notation}

\begin{notation}
  \label{not:Circ}
    If $x \in \R$, $\circ(x)$ denotes the floating-point number closest to $x$ (ties can be decided either way).
  Circled operators such as~$\oplus$ denote correctly rounded floating-point operations.
\end{notation}

We always  have $\circ(x) = x\,(1+\delta)$ and $x = \circ(x)\,(1+\delta')$ with $|\delta|, |\delta'| \le 2^{-\twork}$.
We will also extensively use the \emph{relative error counter} notation $z\rnd{k}$.

\begin{notation}
We write $\widehat z = z \rnd k$ when there exist $\delta_1, \dots, \delta_k$ such that $\widehat{z} = z \prod_{i=1}^k (1+\delta_i)^{\pm 1}$ with $|\delta_i| \le 2^{-\twork}$ for all~$i$.
\end{notation}

Roughly speaking, each arithmetical operation adds one to the relative error counter of a variable. The overall error corresponding to an error counter can  be bounded as follows.

\begin{proposition}
\label{errorsAccumulation}
  Suppose that we can write $\widehat{z} = z\rnd{k}$ and that $k\,2^{-\twork} \le 1/2$. Then $\widehat{z}=z(1+\theta)$ with $|\theta| \le 2k \cdot 2^{-\twork}$.
\end{proposition}

\section{Evaluation of the modified series}

As we shall see in the next section, evaluating the auxiliary function $F$ is fairly straightforward.
The evaluation of~$G$ is more involved.
Indeed, while $\sum G_n x^{3n}$ is well-conditioned as a sum for $x \geq 0$ (this is the whole point of the GMR method), the minimality of the \emph{sequence} $(G_n)$ among the solutions of~\eqref{eq:rec G} implies that its direct recursive computation from the initial values~$G_0$ and $G_1$ is numerically unstable (cf. \cite{Wimp1984}).

There is a standard tool to handle this situation, namely Miller's backward recurrence method~{\cite{BickleyComrieMillerSadlerThompson1952,Wimp1984}}. Miller's method allows one to accurately evaluate the minimal solution $(c_n)$ of a recurrence of the form
\begin{equation}
  a_2(n)\,u_{n+2} + a_1(n)\,u_{n+1} + a_0(n)\,u_n = 0, \quad a_0(n) a_2(n) \neq 0. \label{eq:rec general form}
\end{equation}
The idea is as follows: choose a starting index $R$ and let (arbitrarily) $u_{R+1} = 0$ and $u_R = 1$. Then compute $u_n$ as
\[ -a_0(n)^{-1}\,(a_1(n)\,u_{n+1} + a_2(n)\,u_{n+2}), \quad n=R-1, \dots, 1, 0. \]
It turns out that, for large~$R$, the computed sequence $(u_n)$ is close to a minimal solution of the forward recurrence. Since all minimal solutions are proportional to each other, we recover an approximation of $c_n$ as $c_n \approx (c_0/u_0)\,u_n$.

We use Miller's method to evaluate the minimal solution $(c_n)$ of the normalized recurrence~\eqref{eq:norm rec G}, and we get an approximation of $G(x) = \sum_{n \ge 0} G_n\,x^{3n}$ as $S_N = (c_0/u_0)\,\sum_{n=0}^{N-1} u_n\,x^{3n}/(n!^2)$.
The algorithm is summed up as Algorithm~\ref{algoG}. The rest of this section is devoted to its proof of correctness, i.e., the proof that the value $s$ it returns satisfies $|G(x)-s| \le 3 \cdot 2^{-\ttarget} G(x)$.
\begin{algorithm}[t]
  \KwIn{a target precision $\ttarget \ge 1$, a point $x \geq 0.5$}
  \KwOut{$s$ such that $|G(x)-s| \le 3\cdot 2^{-\ttarget}\,G(x)$}
  Choose $\alpha, \beta, \delta, \gamma$ s.t.
  $\alpha \lesssim 3 \, e^{-1} x^{-3/2}$,
  $\beta \lesssim (2/3)\,\log_2(e)\,x^{3/2}$,
  $\gamma \gtrsim 1/\log_2(20/3)$,
  $\delta \gtrsim (2/3)\,\log_2(e)\,((20/3)^{1/2}-1)\,x^{3/2}$ \;
  $N_0 \gets \max(1, \lceil (3/10)^{1/2}\,x^{3/2}-1 \rceil)$\nllabel{choiceOfN}\;
%
%
%
    Choose $N \ge N_0$ s.t. $\EXP((\alpha\,N)^{2N}) \ge \ttarget + 9 +
    \frac34\EXP(x) - \lfloor \beta \rfloor$\nllabel{choiceOfNlargex}\;
    Choose $R \ge \max(N, (\ttarget + 2 + \delta)\,\gamma)$\;
    Choose $\twork$ s.t. $128\,(N+3)\,2^{-\twork} \le 2^{-\ttarget}$ and $(R+2)\,2^{-\twork}\le 2^{-9}$\nllabel{choiceeOft}\;
  \For{$(a \gets 1, b \gets 0, i \gets R-1; i\ge0; i\gets i-1)$}{
    $c \gets a$\;
    $a \gets (10\otimes a)\, \ominus\, (3i+4)(3i+5)\otimes b \oslash ((i+1)(i+2))$\nllabel{matrixMul1}\;
    $b \gets c$\nllabel{matrixMul2}\;
    \lIf{$i = N-1$}{$s' \gets a$\nllabel{horner1}\;}
    \lElseIf{$i < N-1$}{$s' \gets a \oplus (x^3 \otimes s' \oslash (i+1)^2)$\nllabel{horner2};}  }
  \KwRet $s = (s' \oslash a) \oslash \circ(9\,\Gamma(2/3)^3)$\nllabel{finalDivision}\;
  \caption{Evaluation of $G$}
  \label{algoG}
\end{algorithm}

\begin{proposition}
  \label{prop:N well chosen}
  With $N$ as in Algorithm~\ref{algoG}, the truncation error satisfies $|\sum_{n=N}^{\infty} G_n x^{3n}| \le 2^{-\ttarget}\,G(x)$ for all $x \geq 0.5$.
\end{proposition}
\begin{IEEEproof}
    First, because of line~\ref{choiceOfN} of the algorithm, we have $\frac{3}{10}x^3 \le (N+1)^2$, so that
    \[
      \Bigl| \sum_{n=N}^{\infty} G_n x^{3n} \Bigr|
      \leq
      2G_N x^{3N}
      \leq
      2 \left(\frac{ex^{3/2}}{3N}\right)^{2N}
      \leq
      2 \, (\alpha N)^{-2N}
    \]
    by Lemma~\ref{lem:remainder G} and Corollary~\ref{cor:upper bound on Gn}.
    Line~\ref{choiceOfNlargex} then ensures that 
    \[
    2 \, (\alpha N)^{-2N}
    \le \frac{2^{-\ttarget}}{128} \, e^{\frac{2}{3} x^{3/2}} x^{-3/4}
    \le 2^{-p} G(x),
  \]
  the last inequality coming from Lemma~\ref{lem:bounds G}.
\end{IEEEproof}

There are two sources of error besides the truncation: first, $(u_n)$ is not exactly proportional to $(c_n)$, especially when $n$ is close to $R$. Second, roundoff errors happen during the evaluation of~$(u_n)$.
Rigorous bounds for both sources of error have been proposed by Mattheij and van~der~Sluis~\cite{MvdS}.
We combine them with classical techniques (well-explained, e.g., in~\cite{Chevillard2012}) to choose the starting index $R$ and the working precision~$\twork$ so as to guarantee the final accuracy.

We now recall Mattheij and van~der~Sluis' main result (adapted to our particular case, which simplifies the statement quite a bit).
Consider a recurrence of the form~\eqref{eq:rec general form}. Denote by $(c_n)$ a minimal solution that we wish to evaluate, and let $(d_n)$ be the solution such that $d_0 = d_1 = 1$. Assume that $(d_n)$ is a dominant solution and that the sequences $(c_n)$, $(d_n)$ and $(c_n/d_n)$ are decreasing. Define $H = \frac{d_0}{c_0} \sum_{i=0}^{R-1} \frac{c_i}{d_i}$ and, for $i \le R$,
\[U_i = 
    \begin{pmatrix}
      c_i & c_i\\
      c_{i+1} & c_i \,\frac{d_{i+1}}{d_i}
    \end{pmatrix} \text{ and }
B_i = 
    \begin{pmatrix}
      \frac{-a_1(i)}{a_0(i)} & \frac{-a_2(i)}{a_0(i)}\\
      1 & 0
    \end{pmatrix}.
\]
Let $v_R \in \R^2$ be a column vector, and for $i \le R-1$, let $v_i$ be the result of the floating-point evaluation of $B_i\,v_{i+1}$ at precision~$\twork$. Write $v_i = \transp{(u_i, u_{i+1})}$. If all operations were exact, $(u_i)$ would be the solution of the recurrence such that $\transp{(u_R, u_{R+1})}=v_R$. To take rounding errors into account, we write  $v_i = (B_i + 2^{-\twork} \mathcal{G}_i)v_{i+1}$ for some matrix~$\mathcal{G}_i$ instead. Define $y_R = \transp{(y_{R1}, y_{R2})} = U_R^{-1} v_R$. Let $\Phi_i = \|U_i^{-1}\mathcal{G}_i U_{i+1}\|$,  the matrix norm being subordinate to the $\|\cdot\|_\infty$ norm for vectors, and let $\Phi \ge \max_i \Phi_i$. 

\begin{theorem} \cite[Theorem~4.1]{MvdS}
\label{thm:MvdS}
Provided that the quantities
\[
  \Phi R 2^{-\twork}, \quad
  \frac{c_R\,d_0}{d_R\,c_0} \left|\frac{y_{R2}}{y_{R1}}\right|, \quad
  \frac{\|y_R\|_\infty}{|y_{R1}|}\, (R+H) (1.3 \Phi 2^{-\twork})
\]
are all bounded by~$0.1$, the approximate value $u_i$ computed by Miller's algorithm satisfies
$(c_0/u_0)\,u_i = c_i \, (1+\theta_i)$ for some $\theta_i$ such that $|\theta_i| \le T_i + R_i$, where
\[
  T_i = 1.5 \left(\frac{c_R\,d_i}{c_i\,d_R} + \frac{c_R\,d_0}{c_0\,d_R}\right) \left|\frac{y_{R2}}{y_{R1}}\right|,
  \quad
 R_i = 1.5 \, \frac{\|y_R\|_\infty}{y_{R1}} \, \varepsilon \, (i+2H),
\]
and $\varepsilon = 1.3 \Phi 2^{-\twork}$.
\end{theorem}

Turning back to the special case $c_n = n!^2 G_n$, Theorem~\ref{thm:MvdS} applied to~\eqref{eq:norm rec G} yields the following.
Recall that $\tau = 3/20$.

\begin{corollary}
\label{cor:MvdS to our case}
Set $v_R = \transp{(1, 0)}$ and
\[B_i = 
  \begin{pmatrix}
    10 & -r(i)\\
    1  & 0
  \end{pmatrix}
  \quad
  \text{where}
  \quad
  r(n) = \frac{(3n+4)(3n+5)}{(n+1)(n+2)}.
  \]
Then, in the notation of Theorem~\ref{thm:MvdS}, we have
$T_i \le \tau^{R-i}$ and $R_i \le 76.5(i+4)2^{-\twork} \le 76.5(N+3)2^{-\twork}$
for all $i <N$.
\end{corollary}

Proving this corollary still requires some work. We postpone it for a bit to explore the consequences of this statement.

Observe that lines~\ref{matrixMul1} and~\ref{matrixMul2} of Algorithm~\ref{algoG} are equivalent to a floating-point evaluation of $B_i v_{i+1}$ where $v_{i+1} = \transp{(a, b)}$. Hence, at each loop turn, we have $a = u_i$ just after line~\ref{matrixMul1}. Lines~\ref{horner1} and~\ref{horner2} are a Horner-like evaluation scheme, so that $s' \approx \sum_{i=0}^{N-1} u_i x^{3i}/i!^2$ at the end of the loop. More precisely, assuming $x^3$ is approximated by $\circ(x) \otimes \circ(x) \otimes \circ(x)$ and division by $(i+1)^2$ is performed as two successive divisions by $(i+1)$, an easy induction shows that one can write
\[ s' = \sum_{i=0}^{N-1} \left(u_i \,\frac{x^{3i}}{i!^2}\right)\rnd{9(i+1)}. \]
Line~\ref{finalDivision} adds $3$ to all error counters. The choice of~$\twork$ on line~\ref{choiceeOft} ensures that $(9(N+1)+3)\cdot2^{-t}\leq1/2$.  Using Prop.~\ref{errorsAccumulation}, we conclude that the sum~$s$ returned by Algorithm~\ref{algoG} satisfies
\[ s = \sum_{i=0}^{N-1} \frac{c_0\,u_i}{u_0} \cdot \frac{x^{3i}}{i!^2}\,(1+\mu_i) = \sum_{i=0}^{N-1} G_i\,x^{3i}(1+\mu_i)(1+\theta_i), \]
where $|\mu_i| \le 2 \, (9(i+1)+3) \cdot 2^{-\twork} \le 18(N+3)\,2^{-\twork}$ and $|\theta_i| \le T_i + R_i.$

Since $R \ge N$, we have $T_i \le \tau$ for all $i < N$ by Corollary~\ref{cor:MvdS to our case}. The choice of~$\twork$ also implies $R_i \le 76.5/256$. Altogether, this ensures that $|\theta_i| \le 1$. Writing $(1+\mu_i)(1+\theta_i) = 1+\delta_i$, we get 
\[ |\delta_i| \le 2|\mu_i| + |\theta_i| \le 112.5(N+3)\,2^{-\twork} + \tau^{R-i}
   \le 2^{-\ttarget} + \tau^{R-i},
\]
and therefore
$ \left|s - \sum_{i=0}^{N-1} G_i x^{3i}\right| \le 2^{-\ttarget} G(x) + \tau^{R} G(\tau^{-1/3} x)$.
\begin{lemma}
    For $x > 0.5$, we have $\tau^R\, G(\tau^{-1/3}x) \le 2^{-\ttarget}G(x)$.
\end{lemma}
\begin{IEEEproof}
  It follows from Lemma~\ref{lem:bounds G} (and $\tau < 1$) that
    \[ \frac{G(x)}{G(\tau^{-1/3} x)}
        \ge
        \frac14 \exp\left(\frac{2}{3} x^{3/2} (1-\tau^{-1/2}) \right),
    \]
    and the algorithm ensures that 
    \[ R \log_2 (\tau^{-1}) \ge \ttarget + 2 + \frac23 x^{3/2}(\tau^{-1/2}-1) \log_2(e), \]
    whence $\tau^{R} \le \frac{1}{4} \, 2^{-\ttarget} \, \exp(\frac{2}{3}\,x^{3/2}(1-\tau^{-1/2}))$ and the result.
\end{IEEEproof}

We can now prove the correctness of Algorithm~\ref{algoG}:
\begin{theorem}
  The value $s$ returned by Algorithm~\ref{algoG} satisfies $|G(x)-s| \le 3\cdot 2^{-\ttarget} G(x)$.
\end{theorem}
\begin{proof}
  It follows from Prop.~\ref{prop:N well chosen} and the above discussion, since $|G(x)-s| \le |G(x) - \sum_{i=0}^{N-1} G_i x^{3i}| + |s - \sum_{i=0}^{N-1} G_i x^{3i}|$.
\end{proof}
The remainder of this section is devoted to the proof of Corollary~\ref{cor:MvdS to our case}.
We begin with a crucial lemma.
Let $(d_n)$ be the solution of~\eqref{eq:norm rec G} defined by $d_0 = d_1 = 1$, let $\eta(n)=1/(3n^2)$, and let $r(n)$ be as in Corollary~\ref{cor:MvdS to our case}.
\begin{lemma} \mightbeomitted
  \label{lem:bound dn}
  For all $n \ge 1$, we have $d_{n+1} \le d_n \le (1+\eta(n))\,d_{n+1}$.
\end{lemma}
\begin{IEEEproof}
  We proceed by induction. Since $d_2 = 9/10$,  the property is true for $n=1$. Now, supposing it for an arbitrary~$n$, we get $(9-\eta(n)) d_{n+1} \le 10d_{n+1}-d_n \le 9d_{n+1}$, so
\[ 
  \frac{9-\eta(n)}{r(n)} d_{n+1} \le d_{n+2} \le \frac{9}{r(n)}d_{n+1}.
\]
We conclude by observing that $9/r(n) \le 1$ and $r(n)/(9-\eta(n)) \le 1+\eta(n+1)$ for $n \ge 1$.
\end{IEEEproof}
\begin{corollary} \mightbeomitted
  \label{cor:estim dn}
  For all $n$, we have $0.783 \le d_n \le 1$.
\end{corollary}
\begin{IEEEproof}
  By Lemma~\ref{lem:bound dn}, $(d_n)$ is decreasing and $d_0 = 1$: this proves the right-hand side.
  For $n \geq 100$, Lemma~\ref{lem:bound dn} implies
\[ d_{100} \le d_n\,\prod_{i=100}^{n-1} (1+\eta(i))
  \le d_n \frac{p_{\infty}}{p_{99}}, \quad p_k=\prod_{i=1}^k (1+\eta(i)).
\]
Using the exact value
$p_{\infty} = \frac{\sqrt{3}}{\pi} \sinh(\frac{\pi}{\sqrt{3}})$
\cite[Eq.~4.36.1]{DLMF},
we check that
$d_n \geq d_{100} \, p_{99} \, p_{\infty}^{-1} \geq 0.783$
for $n \geq 100$.
As $(d_n)$ is decreasing, the inequality holds for $n<100$ too.
\end{IEEEproof}

This estimate, combined with Corollary~\ref{cor:positivity} gives almost all we need to check the hypotheses of Theorem~\ref{thm:MvdS}. We use the notation introduced for the statement of the theorem (specialized to the computation of $c_n=n!^2 G_n$ using~\eqref{eq:norm rec G}, with $(d_n)$ as above).

\begin{corollary} \mightbeomitted
  \label{cor:ineqs for MvdS}
  The sequences $(c_n)$, $(d_n)$ and $(c_n/d_n)$ are decreasing. Moreover, the following inequalities hold: $H \le 2$, $|y_{R2}/y_{R1}| \le \frac16$, $\|y_R\|_{\infty}/|y_{R1}| = 1$, and $\det U_i \ge \frac34 c_i^2$.
\end{corollary}
\begin{IEEEproof}
  Corollary~\ref{cor:positivity} shows that $(c_n)$ is decreasing and Lemma~\ref{lem:bound dn} shows that $(d_n)$ is decreasing. Together, they imply
$
  \frac{c_{n+1}}{d_{n+1}} / \frac{c_n}{d_n}
  \leq \tau (1 + \eta(n))
  \leq \frac{1}{5}
$ for $n \geq 1$. We check separately that $c_1/d_1 \leq c_0/d_0$.
Hence $(c_n/d_n)$ is decreasing.

Corollary~\ref{cor:positivity} also implies $c_n \le \tau^n c_0$ for any $n$, and Corollary~\ref{cor:estim dn} shows that $d_0 = 1 \le d_n/0.783$ for any~$n$.
It follows that
\[ H = \frac{d_0}{c_0} \sum_{i=0}^{R-1} \frac{c_i}{d_i} \le \frac{1}{0.783} \sum_{i=0}^{R-1} \tau^i \le 2. \]

We have $d_0/d_1 = 1$, $d_1/d_2 = 10/9$, and
$d_i/d_{i+1} \le 1+1/(3i^2) \le 10/9$
for $i \geq 2$,
and, by definition of $U_R$ and $v_R$,
\setlength{\arraycolsep}{2pt}
\[
  y_R = U_R^{-1} v_R =
  \frac{c_R}{\det U_R}
  \begin{pmatrix}
    d_{R+1}/d_R & -1\\
    -c_{R+1}/c_R &  1
  \end{pmatrix}
  \begin{pmatrix}
    1 \\ 0
  \end{pmatrix}
  =
  \frac{c_R}{\det U_R}
  \begin{pmatrix}
    d_{R+1}/d_R \\ -c_{R+1}/{c_R}
  \end{pmatrix}.
\]
It follows that
$|y_{R2}/y_{R1}| \le \frac{10}{9} \tau = \frac16$,
and hence
$\|y_R\|_{\infty} = |y_{R1}|$.

Finally, from the expression $\det U_i = c_i^2\,(d_{i+1}/d_i - c_{i+1}/c_i)$, we obtain $c_i^{-2}\, \det U_i \geq (\frac{9}{10} - \tau) = \frac34$.
\end{IEEEproof}
\begin{lemma}
A suitable value for $\Phi$ is $39$.
\end{lemma}
\begin{IEEEproof}
  The multiplication $B_i v_i$ is performed on lines~\ref{matrixMul1}
  and~\ref{matrixMul2} of the algorithm.
  Denoting by $a'$ and $b'$ the new values of $a$ and~$b$, we can write $b'=a$ and
  \[ a' = 10a\rnd{2} - r(i) b \rnd{5} = 10a(1+\theta) - r(i) b(1+\theta')\]
  where (since $t \geq 5$ by line~\ref{choiceeOft})
  $|\theta| \le 4\cdot 2^{-\twork}$
  and
  $|\theta'| \le 10\cdot 2^{-\twork}$
  by Prop.~\ref{errorsAccumulation}.
  (This assumes that the multiplication by~$r(i)$ is performed through four successive multiplications and divisions). Therefore, we have
\[
    |\mathcal{G}_i| \le
  \begin{pmatrix}
    4 \cdot 10 & 10\, r(i)\\
    0 & 0
  \end{pmatrix},
\]
where, for matrices, $|A| \le |B|$ means that the inequality holds entrywise. Since moreover
\begin{align*}
  |U_i^{-1}|
  &\le \frac{|c_i|}{|\det U_i|}\,
  \begin{pmatrix}
    d_{i+1}/d_i & 1\\
    c_{i+1}/c_i & 1
  \end{pmatrix} \le \frac{4}{3c_i}\,
  \begin{pmatrix}
    1 & 1\\
    3/20 & 1
  \end{pmatrix},
  \\
  |U_{i+1}|
  &\le c_{i+1}\,
  \begin{pmatrix}
    1 & 1\\
    c_{i+2}/c_{i+1} & d_{i+2}/d_{i+1}
  \end{pmatrix} \le c_{i+1}\,
  \begin{pmatrix}
    1 & 1\\
    3/20 & 1
  \end{pmatrix},
\end{align*}
and $\Phi_i \le |U_i^{-1}| \cdot |\mathcal{G}_i| \cdot |U_{i+1}|$, we get
\[ 
  \Phi_i \le \frac{4}{3}\cdot\frac{3}{20}\,
  \begin{pmatrix}
    40+\frac{3}{2}r(i) & 40+10r(i)\\
    \frac{3}{20}\,(40+\frac{3}{2}r(i)) & \frac{3}{20}\,(40+10r(i))\\
  \end{pmatrix}.
  \]
  Hence $\|\Phi_i\| \le 16 + \frac{23}{10}\,r(i)$. The
result follows because $r(i) \le 10$ for all $i \geq 0$.
\end{IEEEproof}

We can finally prove Corollary~\ref{cor:MvdS to our case}.
\begin{IEEEproof}
    To apply~Theorem~\ref{thm:MvdS}, we need to check that
    \begin{align*}
      \text{(i)}~ & 39\,R 2^{-\twork} \le 0.1,
      &\text{(ii)}~ & \frac{c_R\,d_0}{d_R\,c_0} \le 0.1\,\left|\frac{y_{R1}}{y_{R2}}\right|, \\
      \text{(iii)}~&(R+2)(50.7\cdot 2^{-\twork})\le 0.1.
    \end{align*}
    By definition of $\twork$, we have $R+2 \le 2^{\twork-9}$, so (iii) is satisfied, and (i) follows immediately.
    Corollary~\ref{cor:ineqs for MvdS} combined with the inequalities $d_0/d_R \le 1/0.783$ and $c_R/c_0 \le \tau^R \le \tau$ implies~(ii).

In conclusion, we can apply Theorem~\ref{thm:MvdS}, hence 
\[ T_i = 1.5
\left(\frac{c_R\,d_i}{c_i\,d_R} + \frac{c_R\,d_0}{c_0\,d_R}\right)
\left|\frac{y_{R2}}{y_{R1}}\right| \le \frac{1.5}{0.783}\,\left(\tau^{R-i} +
\tau^{R}\right)\,\frac{1}{6}. \] Therefore, $T_i \le 0.639 \cdot \tau^{R-i} \le \tau^{R-i}$ as announced.
Corollary~\ref{cor:ineqs for MvdS} yields $R_i = 1.5 \, (50.7\cdot 2^{-\twork}) (i+2H) \le 76.5\cdot (i+4)\,2^{\-\twork}$.
\end{IEEEproof}

\section{Evaluation of the auxiliary series}
\label{sec:evaluation of F}

The implementation of the auxiliary series $F$ is much easier than that of~$G$.
We limit ourselves to a sketch of the (fairly standard) algorithm.

A variable $a_0$ is used to successively evaluate $F_0$, $F_3 x^3$, $F_6 x^6$, etc., using the recurrence~\eqref{eq:rec F}. Accordingly, two variables $a_1$ and $a_2$ are used to evaluate the successive values of $F_{3i+1} x^{3i+1}$ and $F_{3i+2} x^{3i+2}$. Each step adds at most $10$ to the relative error counter of each variable. A variable $s$ is used to accumulate the sum as the variables $a_k$ are updated. Therefore, after step $K$, we can write $s = \sum_{i=0}^{3K-1} F_i x^i \rnd{1+10\lfloor i/3 \rfloor+3K-i}$ (the term $3K-i$ representing the errors due to additions). Bounding all errors uniformly, we get
\[ s = \sum_{i=0}^{3K-1} F_i x^i \rnd{10K}. \]

It is easy to see that $q(i) = F_{i+3}/F_i$ decreases for $i\ge
1$, hence the loop can be stopped as soon as (i) $q(3K) x^3 < 1/2$, and (ii)
$a_0, a_1, a_2 < 2^{\EXP(s)-\ttarget-4}$. These conditions ensure that the
remainder $\sum_{i\ge 3K} F_i\,x^{i}$ is bounded by $2(F_{3K} + F_{3K+1} + F_{3K+2}) < 4(a_0+a_1+a_2) < \frac{12}{16} 2^{-\ttarget}\,2^{\EXP(s)}$.

It is clear from Prop.~\ref{prop:rec F} that $F_{n+3} \le 4F_n/n^2$ for
any~$n$, hence (using $n! \approx (n/e)^n$) we have $F_n \approx
(4e^2/n^2)^{n/3}$. This is most likely an overestimation of the true value.
Moreover, we can approximate $F(x)$ by $\frac1{32} x^{-1/2} \exp(\frac{4}{3}
x^{3/2})$ for $x>0.5$. These estimates are used to get a rough overestimation of $K$. The working precision~$\twork$ is then chosen so that $20K\cdot2^{-\twork} \le 2^{-3-\ttarget}$. Hence, we have $|\sum_{i=0}^{3K-1} F_i x^i - s| \le 2^{-3-\ttarget}\,s \le 2^{\EXP(s)-3-\ttarget}\,$.

The initial estimation of the truncation rank is very unlikely to be
underestimated. In the case it would be smaller than the actual truncation rank
decided on-the-fly by the above criterion, this is checked \emph{a posteriori} and, if necessary, the evaluation is run again with an updated working precision.

We conclude by writing $|F(x)-s| \le |F(x)-\sum_{i=0}^{3K-1} F_i x^i| +
|s-\sum_{i=0}^{3K-1} F_i x^i| \le \frac{7}{8} 2^{-\ttarget}\,2^{\EXP(s)} \le 2^{-\ttarget} s$.

\section{Complete algorithm}
\label{sec:implementation}
Assume $\ttarget \ge 3$. From the previous sections, it appears that we computed $\widehat{G}$ such that $\widehat{G} = G(x)(1+\delta_1)$ with $|\delta_1| \le 3\cdot 2^{-\ttarget}$, and we computed $\widehat{F}$ such that $F(x) = \widehat{F}(1+\delta_2)$ with $|\delta_2| \le 2^{-\ttarget}$. Hence $\widehat{G}/\widehat{F} = (G(x)/F(x))(1+\delta_3)$ where $|\delta_3| = |\delta_1 + \delta_2(1+\delta_1)| \le 5 \cdot 2^{-\ttarget}$. Indeed, the division is performed in floating-point arithmetic at precision~$\ttarget$, leading a final result $\widehat{A} = (\widehat{G}/\widehat{F})(1+\delta_4)$ with $|\delta_4| \le 2^{-\ttarget}$. Hence, finally, $\widehat{A} = \Ai(x)(1+\delta_5)$ with $|\delta_5| = |\delta_3 + \delta_4(1+\delta_3)| \le 7\cdot 2^{-\ttarget} \le 2^{-(\ttarget - 3)}.$

We developed a prototype implementation of this algorithm, based on the multiple-precision floating-point library MPFR~\cite{FousseHanrotLefevrePelissierZimmermann2007}.
For simplicity, we supposed $x \geq 1/2$ in the present paper, but our implementation is valid for any $x \geq 0$.

We compared the results with those of the implementation of $\Ai(x)$ available in MPFR, run with a larger precision.
Random tests using thousands of points~$x$ and target precisions~$\ttarget$ yield relative errors smaller than $2^{-(\ttarget - 3)}$ between the result of our implementation and the result of~MPFR, as predicted by theory.

\begin{figure}
  \centerline{\includegraphics[width=6cm]{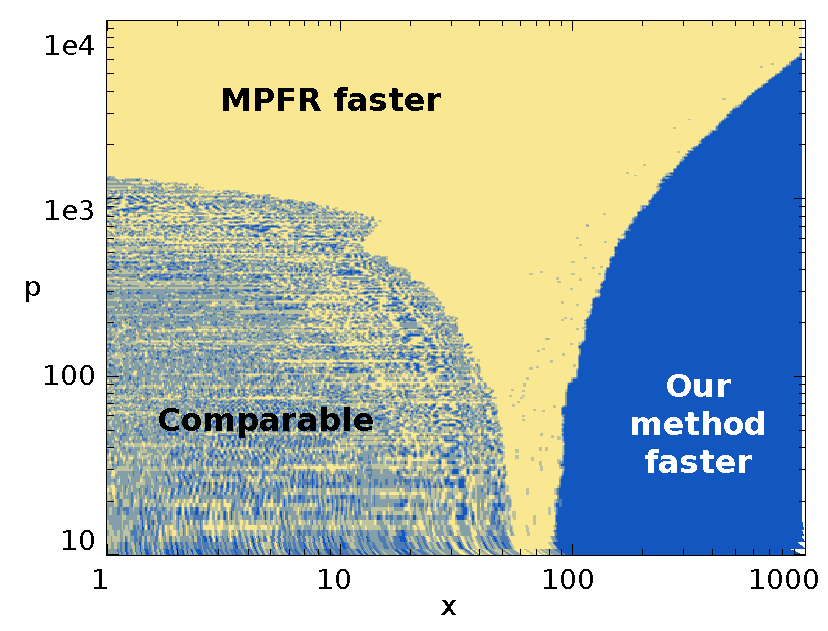}}
  \caption{Fastest method as a function of $x$ and $\ttarget$ (scales are logarithmic). MPFR~3.1.1, GMP~4.3.1, on an Intel Xeon at~2.67GHz.\label{fig:comparison with MPFR}}
\end{figure}

We also ran our implementation and MPFR with the same accuracy to compare their performance (see Fig.~\ref{fig:comparison with MPFR}).
When $x$ is large, our method is faster than MPFR, which uses the Taylor expansion of $\Ai$ at the origin. This is all the benefit of reducing the cancellation: for large $x$, the working precision of MPFR has a large overhead because of the bad condition number of the series.
In contrast, when $\ttarget$~is large compared to~$x$, our implementation pays the cost of evaluating two series instead of one, with little benefit in terms of working precision.

To be completely fair, we should mention that MPFR does not implement the asymptotic expansion of~$\Ai$ at infinity, which should be a better choice than our algorithm for large~$x$ and comparatively small~$\ttarget$. On the other hand, our code currently uses only the naive series summation algorithm, while MPFR implements Smith's baby steps-giant steps technique~\cite{PatersonStockmeyer1973,Smith1989} that is more efficient. There is no theoretical obstacle to using Smith's method with our series: it only obfuscates a little bit the description of the algorithm and the roundoff error analysis. Once we will have implemented it, we will have a more complete picture with three areas, indicating what method between the Taylor series, the asymptotic series and our series is the most efficient, depending on $(x,\ttarget)$.

\section{Outlook: Towards a GMR algorithm?}
\label{sec:conclusion}

As mentioned in the introduction, we see the present work as a case study.
Indeed, in spite of the many technical details that occupy much of the space of this article, we really used few specific properties of the Airy function besides Equation~\eqref{eq:deq Ai}.
Looking back, our analysis essentially relies on the following ingredients.

(i)~\emph{The ability to find auxiliary series.}
The indicator functions used in the GMR method depend only on the behaviour at (complex) infinity of the entire functions they are associated to.
In the case of the solution of a LODE with analytic coefficients, this behaviour is entirely determined by the differential equation along with a finite number of ``asymptotic initial values''~\cite{Wasow1965,vdPutSinger2003}.
Once the indicator function is known, it remains to exhibit appropriate auxiliary series.
Doing this in a truly general way remains an open problem.
Yet, both the original GMR method and our variant apply to many cases, and it is likely that they can be combined and further generalized.

(ii)~\emph{An efficient way to compute their coefficients.}
This seems to be considered a major limitation in the original GMR paper~\cite{GawronskiMullerReinhard2007}.
But, as already mentioned, recurrences with polynomial coefficients automatically exist as soon as both the original function to evaluate and the auxiliary series satisfy differential equations with polynomial coefficients.
Numerical stability is not much of an issue in the case of three-term recurrences, thanks to Miller's method, though many technical details must be settled.
The situation is more complicated in general for recurrences of higher order.
Observe, though, that we proved the minimality of $(G_n)$ using essentially the same asymptotic properties that were exploited by the GMR method in the first place.
The minimality may hence not be fortuitous and might generalize.

(iii)~\emph{Upper and lower bounds on the coefficients and sums of the series
$F$~and~$G$.}
All these bounds were derived, in a pretty systematic way, from the asymptotic expansion of $\Ai$ at infinity combined with Lemma~\ref{lem:asympt approx Ai}.
Bounds similar to that from Lemma~\ref{lem:asympt approx Ai} can themselves often be obtained from a LODE~\cite{Olver1997}.

(iv)~\emph{Roundoff error analyses.}
Our error analyses follow a very regular pattern and could probably be abstracted to a more general case or automated.

In short, most steps of the present study could apparently be performed in a systematic way, starting from Eq.~\eqref{eq:deq Ai} plus a moderate amount of additional information.
Systematizing the GMR method based on this observation seems a promising line of research.
Solutions of LODE with polynomial coefficients are known in Computer Algebra as \emph{D{-}finite}, or \emph{holonomic}, functions.
It would be interesting to isolate a subclass of D-finite functions to which the method applies in a truly systematic way, and attempt to \emph{automate} it at least partially.

\paragraph*{Acknowledgments}

We thank Paul Zimmermann for pointing out the GMR article to us.

\bibliographystyle{amsplain}

\providecommand{\bysame}{\leavevmode\hbox to3em{\hrulefill}\thinspace}
\providecommand{\MR}{\relax\ifhmode\unskip\space\fi MR }
\providecommand{\MRhref}[2]{%
  \href{http://www.ams.org/mathscinet-getitem?mr=#1}{#2}
}
\providecommand{\href}[2]{#2}

\end{document}